\documentclass[a4paper,11pt]{article}

\usepackage{authblk}

\usepackage{amsmath}
\usepackage{amsthm}
\usepackage{amssymb}

\usepackage{algorithmic}
\usepackage{algorithm}
\usepackage{xcolor}
\usepackage{graphicx}
\usepackage{cite}
\usepackage{paralist}

\usepackage[left=1.0in,top=1.0in,right=1.0in,bottom=1.0in,nohead]{geometry}

\theoremstyle{definition}

\newtheorem{definition}{Definition}
\newtheorem{theorem}{Theorem}

\newtheorem{lemma}[theorem]{Lemma}

\theoremstyle{plain}

\newenvironment{ap_lemma}[1]{\par\noindent{\bf Lemma~#1.} \em}{}
\newenvironment{ap_theorem}[1]{\par\noindent{\bf Theorem~#1.} \em}{}

\graphicspath{{figures/}}

\def\MC#1{{\mathcal #1}}
\def\MB#1{{\mathbf #1}}
\def\MBB#1{{\mathbb #1}}

\newcommand{\BIGLR}[3]{{\left#1#3\right#2}}
\newcommand{\BIGP}[1]{{\BIGLR{(}{)}{#1}}}
\newcommand{\BIGBP}[1]{{\BIGLR{\{}{\}}{#1}}}
\newcommand{\CEIL}[1]{{\BIGLR{\lceil}{\rceil}{#1}}}

\newcommand{\BIGC}[1]{{\BIGLR{|}{|}{#1}}}

\newcommand{\OP}[1]{{\operatorname{#1}}}

\long\def\longdelete#1{}

\title{Online Power-Managing Strategy with Hard Real-Time Guarantees} 

\author[*]{Jian-Jia~Chen}
\author[+]{Mong-Jen~Kao}
\author[+]{D.T.~Lee}
\author[*]{Ignaz~Rutter}
\author[*]{Dorothea~Wagner}

\affil[*]{Faculty for Informatics, Karlsruhe Institute of Technology (KIT), Germany. \newline \textit{Email: j.chen@kit.edu, rutter@kit.edu, dorothea.wagner@kit.edu}}

\affil[+]{Research Center for Infor. Tech. Innovation, Academia Sinica, Taiwan. \newline
\textit{Email: mong@citi.sinica.edu.tw, dtlee@ieee.org}
}

\begin{document}

\maketitle

\pagestyle{plain}


\begin{abstract}
We consider the problem of online dynamic power management that provides hard real-time guarantees.
In this problem, each of the given jobs is associated with an arrival time, a deadline, and an execution time, and the objective is to decide a schedule of the jobs as well as a sequence of state transitions on the processors so as to minimize the total energy consumption.
%
%
%
%
%
%
%
In this paper, we examine the problem complexity and provide online strategies to achieve energy-efficiency.
%
First, we show that the competitive factor of any online algorithm for this problem is at least $2.06$.
%
Then we present an online algorithm which gives a $4$-competitive schedule.
When the execution times of the jobs are unit, we show that the competitive factor improves to $3.59$.
At the end, the algorithm is generalized to allow a trade-off between the number of processors we use and the energy-efficiency of the resulting schedule.
%
%
%
\end{abstract}


\section{Introduction}

Reducing power consumption and improving energy efficiency has become an important designing requirement in computing systems. For mobile devices, efficient energy management 
can effectively extend the standby period and prolong 
battery lifetime. For large-scale computing clusters, appropriate power-down mechanism for idling processing units can considerably reduce the electricity bill.

%

In order to increase the energy efficiency, two different mechanisms have been introduced 
to reduce the energy consumed for idling periods.
%
\emph{(1) Power-down Mechanism:}
		When a processor is idling, it can be put into a low-power state, e.g., {sleep} or {power-off}. While the processor consumes less energy in these states, a fixed amount of energy is required to switch the system back to work.
		In the literature, the problem of deciding the sequence of state transitions
		is referred to as \emph{dynamic power management}.
%
%
\emph{(2) Dynamic Speed Scaling:} 
		The concept of dynamic speed scaling refers to the flexibility provided by a processor to adjust its processing speed dynamically. The rate of energy consumption is typically described by a convex function of the processing speed. This feature is also referred to as \emph{dynamic voltage frequency scaling}, following its practical implementation scheme.

%

For systems 
that support the power-down mechanism, Baptiste~\cite{Baptiste06} proposed the first polynomial-time algorithm 
to decide the optimal strategy for turning on and powering off the system for aperiodic real-time jobs with unit execution time.
In a follow-up paper, Baptiste et al.~\cite{DBLP:conf/esa/BaptisteCD07} further extended the result to jobs of arbitrary execution time and reduced the time complexity.
%
When strict real-time guarantees are not required, i.e., deadline misses of jobs are allowed, Augustine et al.~\cite{AugustineIS04} considered systems with multiple low-power states and provided online algorithms.
A simplified version of this problem is also known as \emph{ski-rental}~\cite{DBLP:journals/algorithmica/KarlinMMO94}.

Dynamic speed scaling was introduced to allow computing systems to reach a balance between high performance and low energy consumption dynamically. 
Hence, scheduling algorithms that assume dynamic speed scaling, e.g., Yao et al.~\cite{b:Yao95}, usually execute jobs as slowly as possible while ensuring that timing constraints are met. 
%
When the energy required to keep the processor active is not negligible, however, executing jobs too slow may result in more energy consumption. For most realistic power-consumption functions, there exist a \emph{critical speed}, which is most energy-efficient 
for job execution~\cite{ChenK07,b:Sandy03}.

%

Irani et al.~\cite{b:Sandy03} initiated the study of combining both 
mechanisms. For offline energy-minimization, they presented a $2$-approximation. 
For the online version, they introduced a \emph{greedy procrastinating principle}, which enables any online algorithm for speed scaling without power-down mechanism to additionally support the power-down mechanism.
The idea behind this principle is to postpone job execution as much as possible in order to bundle workload for batch execution.
The usage of job procrastination with dynamic speed scaling for period tasks has later been explored extensively in a series of research~\cite{jjlctes06,b:lee03,ChenK07}.

%

The combination of the power-down mechanism with dynamic speed scaling suggests the philosophy of \emph{racing-to-idle}: Execute jobs at higher speeds and gain longer quality sleeping intervals.
In a recent result, Albers and Antoniadis~\cite{DBLP:conf/soda/AlbersA12} show that the problem of minimizing the energy consumption for speed scaling with a sleep state is NP-hard and provide a $\frac{4}{3}$-approximation.

\vspace{-8pt}
\paragraph{Our Focus and Contribution.} 
In this paper, we examine the problem of online dynamic power management
to minimize the energy consumption.
We present
both lower bounds on the problem complexity and algorithmic results.
%

%

First, we show that the competitive factor of any online algorithm for this problem is at least $2.06$.
%
This shows that this problem is already harder than the ski-rental problem, which has a tight competitive factor of $2$ that is inherited by a couple of online scheduling problems as the only known lower bound~\cite{DBLP:journals/tecs/IraniSG03,DBLP:journals/algorithmica/KarlinMMO94}. 

Second, we present a $4$-competitive online algorithm 
that uses at most two processors for any given set of jobs known in advance to be 
packable on a processor.
When the execution times of the jobs are unit, we show that the competitive factor improves to $3.59$.
%
Then, we generalize our algorithm for a prescribed collection of job streams to allow a trade-off between the number of processors we use and the energy-efficiency of the resulting schedule.

%

%

%

Note that, our assumption on the input job set is crucial in the sense that
packing the jobs using a given number of processors is known to be a long-standing difficult problem even for the offline case~\cite{Chuzhoy:2004:MMS:1032645.1033163,Chuzhoy:2009:RMJ:1616497.1616504,springerlink:10.1007/1-4020-8141-318,Garey:1979:CIG:578533}, and for the online version only very special cases were studied~\cite{Kao2012-mmjs,springerlink:10.1007/1-4020-8141-318}. 

%

Due to the space limit, some technical details and proofs are provided in the appendix for further reference.

\section{Notations and Problem Definition}

In this section, we provide definitions to the scheduling model assumed in this paper, followed by a formal problem definition.


\paragraph{The Jobs.} 
Each job
$j$ is associated with three parameters, namely, the arrival time $a_j$, the execution time $c_j$, and the deadline $d_j$.
The arrival time of a job is the moment it arrives to the system and becomes ready for execution.
The execution time is the amount of time it requires to finish its task,
and the deadline is the latest moment 
at which the task must be completed.
%
We assume that $c_j$ and $d_j$ are both known at the moment when $j$ arrives to the system.

For notational brevity, for any job $j$, we use a triple $j=(a_j, d_j, c_j)$ to denote the corresponding parameters of $j$. 
We call a job $j$ a unit job if $c_j=1$ and we write $j=(a_j, d_j)$. Moreover, a job $j$ is said to be \emph{urgent} if
$c_j = d_j - a_j$.


\paragraph{Energy Consumption.}
In this paper, processors for executing jobs are assumed to have three states, namely
\emph{busy, standby,} and \emph{off}. 
When a processor is off, it cannot execute jobs and consumes a negligible amount of energy. 
Switching a processor from off to other states requires $E_w$ units of energy.
%
A processor is in busy state when executing a job. The amount of energy it consumes per unit of time when busy is denoted by $\psi_{b}$.
When a processor is in standby, it consumes $\psi_{\sigma}$ amount of energy per unit of time.
We assume $\psi_{\sigma} \leq \psi_{b}$.
%
For convenience, we use the terminology ``turning on'' and ``turning off'' to denote the transition between the off state and other states.

%

Provided the above notion, the \emph{break-even time}, denoted by $\MC{B}$, is defined as
${E_w} / {\psi_{\sigma}}$. Literally, this corresponds to the amount of time a processor has to stay in standby in order to have the same energy consumption as a turn-on operation.
Break-even time is an important concept that has been
widely used for ski-rental-related problems~\cite{DBLP:journals/algorithmica/KarlinMMO94} and dynamic power management algorithms in the literature, e.g.,~\cite{DBLP:journals/tecs/IraniSG03,b:Sandy03,DBLP:journals/rts/HuangSCTB11}. 






\paragraph{Job Scheduling.}
Let $\MC{J}$ be a set of jobs. A schedule $\MB{S}$ for $\MC{J}$ on a set of processors $\MC{M}$ is to decide for each processor $m \in \MC{M}$:
		(1) a set $\MC{I}_m$ of time intervals during which processor $m$ is turned on, and
%
		(2) a function $\OP{job}_m(t)\colon \MBB{R}^+ \rightarrow \MC{J}$ of time indicating the job to occupy processor $m$ at time $t$.
%
The schedule $\MB{S}$ is said to be \emph{feasible} if for each job $j \in \MC{J}$, there exist a processor $m \in \MC{M}$ such that
$$\sum_{\OP{I} \in \MC{I}_m}\int_{\OP{I} \cap [a_j,d_j]} \delta\BIGP{\OP{job}_m(t), j}\cdot\OP{d}t \ge c_j,$$
where $\delta(x,y) = 1$ if $x=y$ and $\delta(x,y) = 0$ otherwise.
%
%
%
%
The energy consumption of the schedule $\MB{S}$, denoted $\OP{E}(\MB{S})$, is hence
$$\OP{E}(\MB{S}) = \sum_{m\in \MC{M}}\BIGP{ E_w\cdot\BIGC{\MC{I}_m} + \int_{\OP{I} \in \MC{I}_m} \BIGP{\enskip E_b - \delta(\OP{job}_m(t),\phi)\cdot(E_b-E_s) \enskip} \cdot \OP{d}t }.$$
The goal of 
the \emph{Power-Minimizing Scheduling Problem} is to find a feasible schedule $\MB{S}$ such that $\OP{E}(\MB{S})$ is minimized.
%

%



%
In this paper, we consider the case where the jobs are arriving to the system dynamically in an online setting, i.e., at any time $t$, 
we can only see the jobs whose arrival times are less than or equal to $t$, 
and the scheduling decisions have to be made without prior knowledge on future job arrivals.
To be more precise, let $\MC{J}$ be the input job set and
$\MC{J}(t) = \BIGBP{j \colon j\in \MC{J}, a_j \le t}$
be the subset of $\MC{J}$ 
that consists of the jobs whose arrival times are no greater than $t$.

\begin{definition}[Online Power-Minimizing Scheduling]
For any given set $\MC{J}$ of jobs,
the \emph{online power-minimizing scheduling} problem is to compute a feasible schedule
such that the energy consumed up to time $t$ is small with respect to $OPT(\MC{J}(t))$, where $OPT(\MC{J}(t))$ is the energy consumed by an optimal schedule of $\MC{J}(t)$, for any $t \ge 0$.
\end{definition}

%
%





%

\paragraph{The Schedulability of the Jobs.}

Chetto et al.~\cite{DBLP:journals/ipl/ChettoC89} studied the schedulability of any given set of jobs and proved the following lemma.
\begin{lemma}[Chetto et al.~\cite{DBLP:journals/ipl/ChettoC89}] \label{lemma_offline_condition}
For any set $\MC{J}$ of jobs, $\MC{J}$ can be scheduled on one processor using the \emph{earliest-deadline-first (EDF)} principle, which always selects the job with earliest deadline for execution at any moment, if and only if the following condition holds:
\begin{equation}
\text{For any time interval $(\ell, r)$, we have} \quad \sum_{j\colon j\in \MC{J}, \ell \le a_j, d_j \le r} c_j \le r - \ell. \label{eq-edf}
\end{equation}
\end{lemma}

%
Note that, it is well-known that, for any set $\MC{J}$ of jobs, if there exists a feasible schedule for $\MC{J}$ that uses only one processor, then the EDF principle is guaranteed to produce 
a feasible schedule~\cite{DBLP:conf/ifip/Dertouzos74}.
Therefore, Condition~(\ref{eq-edf}) gives a necessary and sufficient condition for any set of jobs to be able to be packable on a processor.

%



\section{Problem Lower Bound}

In this section, we prove a lower bound of $2.06$ on the competitive factor of any online algorithm by designing an online adversary $\MC{A}$ 
that observes the behavior of the scheduling algorithm to determine the forthcoming job sequence.
%

%

Let $\Pi$ be an online scheduling algorithm for this problem. 
%
%
We set 
$\psi_{b} = \psi_{\sigma} = \psi=1$ and $E_w = k$, where $k$ is an integer chosen to be sufficiently large. Hence the break-even time $\MC{B}$ is also $k$. 
%
Without loss of generality, we assume 
a length of the minimum tick 
of the system to be $\epsilon_0$, 
 which we further assume to be $1$.
%
We define a \emph{monitor} operation of the adversary $\MC{A}$ as follows.

\begin{definition}
When $\MC{A}$ {\bf monitors}
$\Pi$ during time interval $[t_0,t_1]$, it checks if $\Pi$ keeps at least one processor on between time $t_0$ and $t_1$. If $\MC{S}$ turns off all the processors at some point $t$ between $t_0$ and $t_1$,
then $\MC{A}$ releases an urgent job of length $\epsilon_0$ immediately at time $t+\epsilon_0$, forcing $\Pi$ to turn on at least one processor to process it.
If $\Pi$ keeps at least one processor on during the monitored period, then $\MC{A}$ does nothing.
\end{definition}

%

Let $x$, $\eta$, and $\chi$, where $0\le x\le \frac{2}{5}$, be three non-negative parameters to be chosen carefully.
The online adversary works as follows.  At the beginning, say, at time
$0$, $\MC{A}$ releases a unit job $(0,\MC{B},\epsilon_0)$ and observes the
behavior of $\Pi$.  Let $t$ be the moment at which $\Pi$
schedules this job to execute. 
Since $\Pi$ produces a feasible schedule, we know that $0 \leq t \leq \MC{B-1}$.
We have the following two cases.

\smallskip

\noindent\textbf{Case(1):}
	If $0 \leq t \le \BIGP{\frac{1}{2}-x}\MC{B}$, then $\MC{A}$ monitors $\Pi$ from time $t$ to $\frac{3}{2}\MC{B}$.
%
%
\\
\textbf{Case(2):}
	If $j$ is not executed till $\BIGP{\frac{1}{2}-x}\MC{B}$, $\MC{A}$ releases	$\BIGP{\frac{1}{2}+x}\MC{B}-\epsilon_0$ unit jobs with absolute deadline $\MC{B}$ at time $\BIGP{\frac{1}{2}-x}\MC{B}+\epsilon_0$.
	As a result, 
	the online algorithm is forced to wake up at least 
	two processors to meet the deadlines of the jobs.
	The adversary continues to monitor $\Pi$ until time $\BIGP{\frac{3}{2}+\eta}\MC{B}$.
	If no urgent unit jobs have been released till time $\BIGP{\frac{3}{2}+\eta}\MC{B}$, $\MC{A}$ terminates.
	Otherwise, it monitors $\Pi$ for another $\chi\MC{B}$ units of time until $\BIGP{\frac{3}{2}+\eta+\chi}\MC{B}$.

%
\smallskip

Let $\MC{E}(\Pi)$ and $\MC{E}(\MC{O})$ denote the energy consumed by
algorithm $\Pi$ and an offline optimal schedule on the input
sequence generated by $\MC{A}$, respectively.
%
The following lemmas give a lower bound on the ratio $\MC{E}(\Pi) / \MC{E}(\MC{O})$ for each of the aforementioned cases, expressed as a function of $k$, $x$, $\eta$, and $\chi$.

\begin{lemma} \label{lemma:case1-lower}
If the first job released by $\MC{A}$ is scheduled at time $t$ with $0 \leq t \le \BIGP{\frac{1}{2}-x}\MC{B}$, then we have
\vspace{-18pt}
\begin{align*}
\hspace{96pt}
\frac{\MC{E}(\Pi)}{\MC{E}(\MC{O})} \ge 2+\frac{1}{2}x-{O}\BIGP{\frac{1}{k}}, \quad \text{for any $0 \le x \le \frac{2}{5}$}.
\end{align*}
\end{lemma}

\begin{lemma} \label{lemma:case2-lower}
If the first job released by the adversary is not executed until time $\BIGP{\frac{1}{2}-x}\MC{B}$, then for any $x,\eta,\chi \ge 0$ we have
$$\frac{\MC{E}(\Pi)}{\MC{E}(\MC{O})} \ge \min\BIGBP{\frac{3+x+\eta}{\frac{3}{2}+x},\frac{4+x+\eta+\chi}{2+x+\eta},\frac{5+x+\eta+\chi}{2+x+\eta+\chi}}-{O}\BIGP{\frac{1}{k}}.$$
\end{lemma}
By combining Lemma~\ref{lemma:case1-lower} and Lemma~\ref{lemma:case2-lower}
we obtain the following theorem. 
%

\begin{theorem} \label{thm_sparse_lower_bound}
The competitive factor of any online algorithm for the online
power-minimizing 
scheduling problem is at least $2.06$.
\end{theorem}

\section{Online Scheduling}

We have seen in the last section that, in order to 
bundle the execution of the jobs while providing strict real-time guarantees, additional 
processors
are necessary compared to 
those required by an optimal offline schedule.
%

\smallskip

In this section, we first assume that Condition~(\ref{eq-edf}) from Lemma~\ref{lemma_offline_condition} holds for the input set of jobs and present an online strategy that gives an energy-efficient schedule using at most two processors.
%
In \S~\ref{sec-trade-off}, we generalize our algorithm for a prescribed collection of job streams, each of which delivers a set of jobs satisfying Condition~(\ref{eq-edf}), to allow a trade-off between the number of processors we use and the energy-efficiency of the resulting schedule.


\smallskip

%
We begin our discussion with a review on the commonly used approaches and their drawbacks in our problem model.

\paragraph{Common Approaches and Bad Examples.}

A commonly used approach to bundling the workload is to delay the execution of the jobs as long as possible until no more space for delaying the job execution is left, followed by using the earliest-deadline-first principle to schedule the jobs, e.g., the algorithm due to Irani et al.~\cite{b:Sandy03}.
When there is no more job to execute, the ski-rental problem and related scheduling problems~\cite{DBLP:journals/algorithmica/KarlinMMO94} suggest that we stay in standby for $\MC{B}$ amount of time before turning off the processor.
%
Let $\MC{L}$ denote this approach.
%


%

\begin{lemma} \label{lemma_ltr_lower_bound} 
The competitive factor of $\MC{L}$ can be arbitrarily large.
Furthermore, even when $c_j = 1$ for all $j \in \MC{J}$, the competitive factor of $\MC{L}$ is still at least $6$.
\end{lemma}

%

%
The above lemma shows that, when the jobs can have arbitrary execution times, packing of the jobs should be done more carefully.
Furthermore, even when we have $c_j = 1$ for all $j \in \MC{J}$,
there is still room for improvement on the competitive factor.


\paragraph{Our Main Idea.}

The examples provided in Lemma~\ref{lemma_ltr_lower_bound}
give a rough idea on the drawbacks of $\MC{L}$, which are twofold: (1) Scheduling jobs on different processors using a global priority queue can easily result in deadline misses. (2) Blindly delaying the jobs can lead to a less energy-efficient schedule.

%

\smallskip

The former one is more straightforward to deal with. By suitably partitioning the job set, the feasibility can be assured by our assumption on Condition~(\ref{eq-edf}). 
For the latter problem, we introduce the concept of \emph{energy-efficient anchors} for the jobs in order to determine the appropriate timing to begin their execution.

\subsection{Our Algorithm}

We define some notations to help present our online algorithm $\MC{S}$ and the analysis that follows.
%
%
Let $\MC{J}$ be the input set of jobs, and recall that $\MC{J}(t)$ 
is the subset 
of jobs whose arrival times are smaller than or equal to $t$.
%

\smallskip

For any $t, t^\dagger$ with $0\le t\le t^\dagger$, let $\MB{Q}(t)$ be the subset of $\MC{J}(t)$ that contains the jobs 
that have not yet finished their execution, 
%
and let $\MB{Q}(t,t^\dagger)$ 
be the subset of $\MB{Q}(t)$ containing those jobs whose deadlines are smaller than or equal to $t^\dagger$.
Note that, by definition, we have $\MB{Q}(t,t^\dagger) \subseteq \MB{Q}(t) \subseteq \MC{J}(t) \subseteq \MC{J}$.
For notational brevity, let $c'_j(t)$ denote the remaining execution time of job $j$ at time $t$, and let $W(t) = \sum_{j \in \MB{Q}(t)}c'_j(t)$ and $W(t, t^\dagger)
= \sum_{j \in \MB{Q}(t,t^\dagger)}c'_j(t)$ denote the total remaining execution time of the jobs in $\MB{Q}(t)$ and $\MB{Q}(t,t^\dagger)$, respectively. 
Furthermore, we divide $\MB{Q}(t)$ into two subsets according to the arrival times of the jobs. For any $t, t^*$ with $0\le t^* \le t$,
let $\MB{Q}^{t^*}_{proc}(t)$ be the subset of $\MB{Q}(t)$ containing the jobs whose arrival times are less than $t^*$, and 
let $\MB{Q}^{t^*}_{forth}(t) = \MB{Q}(t) \backslash \MB{Q}^{t^*}_{proc}(t)$.

\smallskip

Let $\lambda$, $0 \le \lambda \le 1$, be a constant to be determined 
later.
For each job $j \in \MC{J}$, we define a parameter $h_j$ to be $\max\BIGBP{a_j, d_j - \lambda\MC{B}}$.
The value $h_j$ is referred to as the \emph{energy-efficient anchor} for job $j$.

\smallskip

Let $M_1$ and $M_2$ denote the two processors which our algorithm $\MC{S}$ will manage.
We say that the system is \emph{running}, if at least one processor is executing a job.
The system is said to be \emph{off} if all processors are turned off. 
Otherwise, the system is said to be in \emph{standby}.
%
%
%
%
%
During the process of job scheduling, our algorithm $\MC{S}$ maintains an \emph{urgency flag}, which is initialized to be \emph{false}.
%
At any time $t$, $\MC{S}$ proceeds as follows.
\begin{enumerate}[(A)]
\item {
Conditions for switching on the processors:}

	\begin{enumerate}
	\item {\bf If} the system is \emph{off} and there exists some $j \in \MB{Q}(t)$ such that $h_j \ge t$, {\bf then} turn on processor $M_1$.
	
	\item {\bf If} the \emph{urgency flag} is \emph{false} and there exists some $t^\dagger$ with $t^\dagger > t$ such that $W(t, t^\dagger) \geq t^\dagger - t$, {\bf then} 
	
		\begin{compactitem}
			\item
				turn on $M_1$ if it is \emph{off}, 
				
			\item
				turn on $M_2$, set $t^*$ to be $t$, and set the \emph{urgency flag} to be \emph{true}.
		\end{compactitem}
		
	
	\end{enumerate}

\item {
To handle the job scheduling:}

	\begin{enumerate}
	\item {\bf If} the \emph{urgency flag} is \emph{true}, {\bf then} use the earliest-deadline-first principle to schedule jobs from $\MB{Q}^{t^*}_{proc}(t)$ on $M_1$ and jobs from $\MB{Q}^{t^*}_{forth}(t)$ on $M_2$.
    
	\item {\bf If} the \emph{urgency flag} is \emph{false} and the system is not \emph{off}, {\bf then} use the EDF principle to schedule jobs from $\MB{Q}(t)$ on the processor that is \emph{on}.
	\end{enumerate}

\item {
Conditions for turning off the processors:}

	\begin{enumerate}
	\item {\bf If} the \emph{urgency flag} is \emph{true} and $\MB{Q}^{t^*}_{proc}(t)$ becomes empty,
	{\bf then} turn off $M_1$ and set the \emph{urgency flag} to be \emph{false}.
	
	\item {\bf If} the \emph{urgency flag} is \emph{false}, the system is \emph{standby}, and $t-t_1 \geq \MC{B}$, where $t_1$ is the time processor $M_1$ was turned on, {\bf then} turn off all processors.
	\end{enumerate}
	
\end{enumerate}

Note that, $M_1$ and $M_2$ can both be on only when the \emph{urgency flag} is 
\emph{true}.

\subsection{The Analysis} \label{sec-sparse-analysis}

Let $I_1=(\ell_1, r_1), I_2=(\ell_2,r_2), \ldots, I_\kappa=(\ell_\kappa,r_\kappa)$, where $r_i<\ell_j$ for $1\le i<j\le \kappa$, be the set of time intervals during which the system is either \emph{running} or \emph{in standby}. 
We also refer these intervals to as the awaken-intervals.
For 
ease of presentation, let $I_0 = (0,0)$ be a dummy awaken-interval.

\smallskip

\begin{lemma} \label{lemma-feasibility}
Provided that Condition~(\ref{eq-edf}) from Lemma~\ref{lemma_offline_condition} holds for the input jobs, algorithm $\MC{S}$ always produces a feasible schedule.
\end{lemma}

Below, we bound the competitive factor of $\MC{S}$.
Let $\MC{P}_i$ be the number of times $\MC{S}$ turns on a processor
in $I_i$. We have the following lemma regarding $\MC{P}_i$.

\begin{lemma}
\label{lemma-awaken-interval-property}
For each $i$, $1\le i\le \kappa$, we have
\begin{compactitem}
	\item
		$\MC{P}_i \le 2$.
		
		\smallskip
		
	\item
		If $\MC{P}_i = 2$, then the amount of workload that arrives after $r_{i-1}$ and has to be done before $r_i$ is at least $\lambda\MC{B}$.
\end{compactitem}
\end{lemma}

Let $\MB{J}_i$ denote the set of the jobs executed in awaken-interval
$I_i$ in $\MC{S}$.
As a consequence to Lemma~\ref{lemma-awaken-interval-property}, we have $\sum_{j\in \MB{J}_i} c_j\ge \lambda\MC{B}$ for each $I_i$ with $\MC{P}_i=2$, $1\le i\le \kappa$.
Let $\MC{O}$ be an optimal offline uni-processor schedule.
Below, we 
relate our online schedule to the optimal offline schedule $\MC{O}$. 
Suppose that 
$U_i = (\ell_{i,opt}, r_{i,opt})$, $1\le i\le m$, is the set of intervals for which $\MC{O}$ keeps the system off. 
We also refer these intervals to as the \emph{sleep intervals} of $\MC{O}$.

\begin{lemma} \label{lemma_sparse_sleep_vs_working} 
For any $i$, $1\le i\le \kappa$, if 
both $r_{i-1}$ and $\ell_i$ are contained in some sleep interval $U_u = (\ell_{u,opt}, r_{u,opt})$ of $\MC{O}$,
then $\MC{P}_i = 1$ and $r_i - r_{u,opt} \ge \BIGP{1-\lambda}\MC{B}$.
\end{lemma}

\begin{lemma} \label{lemma_sparse_sleep_n_awaken}
Each sleep interval of $\MC{O}$ intersects with at most two awaken-intervals of $\MC{S}$. In particular, each of $U_1$ and $U_m$ intersects with at most one.
\end{lemma}

%

\begin{figure*}[t]
\centering
\includegraphics[scale=0.9]{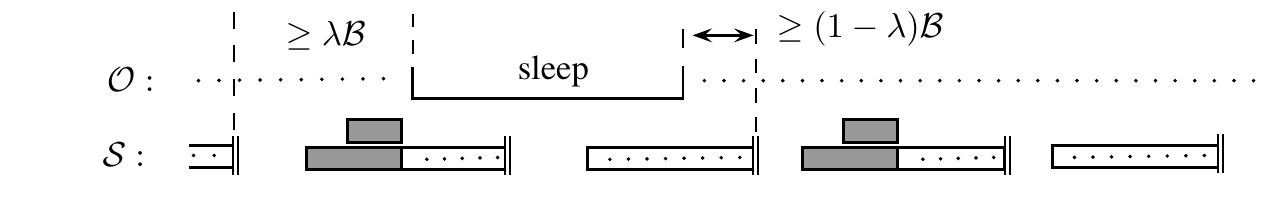}
\vspace{-10pt}
\caption{
The relative positions between sleep intervals of $\MC{O}$ and awaken-intervals of $\MC{S}$.
}
\label{fig_uni_energy_charging}
\end{figure*}

The above two lemmas characterize the relative positions between the sleep intervals of $\MC{O}$ and the awaken intervals they intersect with. See also Fig.~\ref{fig_uni_energy_charging} for an illustration.
For 
simplicity,
we 
implicitly concatenate $U_1$ and $U_m$, i.e., we assume that the left end of $U_1$ is connected to the right end of $U_m$ for the rest of this section.

Below we upper-bound the energy consumption of $\MC{S}$ by the energy consumption of $\MC{O}$. 
In order to help present the analysis,
we define the following notations for 
the energy consumption of the system in $\MC{S}$ and in $\MC{O}$:
(a) $\MC{E}_{w,\sigma}$: energy for waking up and standby, 
(b) $\MC{E}_{w}$: energy for waking up, 
(c) $\MC{E}_{b}$: energy for executing jobs, and
(d) $\MC{E}_{b,\sigma}$: energy for executing jobs and standby. 

\smallskip

We use $\MC{E}_{w,\sigma}(\MC{S}, t_1, t_2)$, and
$\MC{E}_{b}(\MC{S}, t_1, t_2)$ to denote the corresponding
energy $\MC{S}$ consumes between time $t_1$ and $t_2$. Furthermore,
$\MC{E}_{b,\sigma}(\MC{O}, t_1, t_2)$, and
$\MC{E}_{w}(\MC{O}, t_1, t_2)$ are the corresponding
energies $\MC{O}$ consumes between time $t_1$ and $t_2$.
%

\begin{lemma} \label{lemma_sparse_energy_bound}
  For each awaken-interval $I_i$, $1\le i\le \kappa$ and $\lambda\leq 1$, 
  \begin{align}
    \label{eq:energy-pi=2}
    \MC{E}_{w,\sigma}(\MC{S},\ell_i,r_i) \le \psi_{\sigma}\BIGP{\sum_{j\in \MB{J}_i} c_j} + (3-\lambda)E_w. \\
    \label{eq:energy-pi=1}
    \text{Furthermore, if } \MC{P}_i = 1,\text{ then } \MC{E}_{w,\sigma}(\MC{S},\ell_i,r_i) \le 2E_w.
  \end{align}
\end{lemma}

%
%
%

\noindent
By combining all the above analysis in Lemmas
\ref{lemma-awaken-interval-property}, \ref{lemma_sparse_sleep_vs_working},
\ref{lemma_sparse_sleep_n_awaken}, and
\ref{lemma_sparse_energy_bound} for the properties between the optimal
offline schedule and the schedule derived from $\MC{S}$, we prove an upper bound 
on the energy consumption $\MC{E}_{w,\sigma}(\MC{S}, 0, \infty)$ in the following lemma.

\begin{lemma} \label{lemma_sparse_energy_ratio}
When $2-\sqrt{3} \le \lambda \le 1$, we have 
\begin{align*}
 \MC{E}_{w,\sigma}(\MC{S},0,\infty) 
 \le \BIGP{4-\lambda}\cdot\MC{E}_{b,\sigma}(\MC{O},0,\infty) + 4\cdot\MC{E}_{w}(\MC{O},0,\infty).
\end{align*}
\end{lemma}

\begin{proof}
Consider each of the following three exclusive cases of an awaken-interval $I_i$, for $1\le i\le \kappa$.

\smallskip

\noindent	\textbf{Case 1:}
      Interval $I_i$ does not intersect with any sleep interval of
      $\MC{O}$. Therefore, in the optimal schedule $\MC{O}$, the processor is not 
      off
      during time interval $(\ell_i,r_i)$.
      By our algorithm design, the length of an
      awaken-interval is at least $\MC{B}$. 
      Hence
      $\MC{E}_{b,\sigma}(\MC{O},\ell_i,r_i) \ge \MC{B}\cdot
      \psi_{\sigma} = E_w$.  By Lemma \ref{lemma_sparse_energy_bound},
      when $\lambda \leq 1$, we have
      \vspace{-4pt}
      \begin{equation}
        \label{eq:case1}
       \MC{E}_{w,\sigma}(\MC{S},\ell_i,r_i) \le \psi_{\sigma}\cdot\sum_{j\in \MB{J}_i}
      c_j+\BIGP{3-\lambda}\cdot\MC{E}_{b,\sigma}(\MC{O},\ell_i,r_i).
      \end{equation}

\vspace{-6pt}
\noindent \textbf{Case 2:}
		Interval $I_i$ intersects with a sleep interval $U_u = (\ell_{u,opt},r_{u,opt})$ and 
		$r_i$ is contained in interval  $U_u$,
		We have two subcases:
		(a) If $\MC{P}_i=1$, by Eq.~(\ref{eq:energy-pi=1}) in Lemma~\ref{lemma_sparse_energy_bound},
        we have $\MC{E}_{w,\sigma}(\MC{S},\ell_i,r_i) \le 2E_w =
        2\MC{E}_{w}(\MC{O},\ell_{u,opt},r_{u,opt})$, in which the
        equality comes by the definition of a sleep interval.
		%
		(b) If $\MC{P}_i = 2$, to provide an upper bound to
        Eq.~(\ref{eq:energy-pi=2}) in
        Lemma~\ref{lemma_sparse_energy_bound}, we use the energy
        consumption $\MC{E}_{w}(\MC{O},\ell_{u,opt},r_{u,opt})$ in
        the optimal schedule to bound $2E_w$, and 
        the energy
        consumption $\MC{E}_{b,\sigma}(\MC{O},r_{i-1},\ell_{u,opt})$
        in the optimal schedule to bound $(1-\lambda)E_w$. 
%
        We know that, when $2-\sqrt{3}\le
        \lambda\le 1$, we have
        \vspace{-6pt}
        \begin{equation}
          \label{eq:lambda}
          1-\lambda \leq \lambda(3-\lambda).          
          \vspace{-6pt}
        \end{equation}
        By Lemma~\ref{lemma-awaken-interval-property}, in $\MC{O}$, the
        workload that has to be done between $r_{i-1}$ and
        $\ell_{u,opt}$ is lower-bounded by $\lambda\MC{B}$. Since $r_i
        \geq \ell_{u,opt}$,  the  workload that has to be done between $r_{i-1}$ and
        $r_i$ is also lower-bounded by $\lambda\MC{B}$,  and hence,
        \vspace{-6pt}
        \begin{equation}
          \label{eq:lambda22}
          \MC{E}_{b,\sigma}(\MC{O},r_{i-1}, r_{i}) \ge
          \psi_{\sigma}\cdot\lambda\MC{B} = \lambda E_w,
          \vspace{-6pt}
        \end{equation}
        Hence, by Eq.~(\ref{eq:lambda}) and Eq.~(\ref{eq:lambda22}),
        we have
		$\BIGP{1-\lambda}E_w \le  \BIGP{3-\lambda}\cdot\MC{E}_{b,\sigma}(\MC{O},r_{i-1},r_i).$
%
		By rephrasing Eq.~(\ref{eq:energy-pi=2}) in
        Lemma~\ref{lemma_sparse_energy_bound}, for both subcases, 
        we have
        \vspace{-6pt}
		\begin{align}
		& \MC{E}_{w,\sigma}(\MC{S},\ell_i,r_i) \le \psi_{\sigma}\sum_{j\in \MB{J}_i} c_j + \BIGP{1-\lambda}\cdot E_w + 2\cdot E_w \nonumber\\
		\le \enskip & \psi_{\sigma}\cdot\sum_{j\in \MB{J}_i} c_j + \BIGP{3-\lambda}\cdot\MC{E}_{b,\sigma}(\MC{O},r_{i-1},r_i) 
		+ 2\cdot\MC{E}_{w}(\MC{O},\ell_{u,opt},r_{u,opt}).\label{eq:case2}
		\end{align}
		%

\noindent \textbf{Case 3:}
		$I_i$ intersects exactly with one sleep interval $U_u = (\ell_{u,opt}, r_{u,opt})$ and $\ell_i$ is contained in interval  $U_u$,
		then, by Lemma~\ref{lemma_sparse_sleep_vs_working}, we have $\MC{P}_i = 1$ and therefore
		\vspace{-6pt}
        \begin{equation}
          \label{eq:case3}
          \MC{E}_{w,\sigma}(\MC{S},\ell_i,r_i) \le 2\cdot E_w \le 2\cdot\MC{E}_{w}(\MC{O},\ell_{u,opt},r_{u,opt}). 
          \vspace{-6pt}
        \end{equation}
%
Provided the individual upper bounds for each awaken interval in Eq.~(\ref{eq:case1}), Eq.~(\ref{eq:case2}), and Eq.~(\ref{eq:case3}), we combine them to get an overall upper bound for $\MC{E}_{w,\sigma}(\MC{S},0,\infty)$, which is the summation over each awaken interval, $\sum_{1\le i\le \kappa}\MC{E}_{w,\sigma}(\MC{S},\ell_i,r_i)$.

First, since the awaken intervals are mutually disjoint, the item $\MC{E}_{b,\sigma}(\MC{O},r_{i-1},r_i)$ is counted at most once for each $i$, $1\le i\le \kappa$.
Note that, $\sum_{1\le i\le \kappa}\MC{E}_{b,\sigma}(\MC{O},r_{i-1},r_i) \le \MC{E}_{b,\sigma}(\MC{O},0,\infty)$.
%
Second, by Lemma \ref{lemma_sparse_sleep_n_awaken}, a sleep
interval $U_u$ intersects with at most two awaken-intervals of $\MC{S}$.
Hence, $\MC{E}_{w}(\MC{O},\ell_{u},r_u)$ is counted at most
twice. 
Therefore, when $2-\sqrt{3} \le \lambda \le 1$, we have
\vspace{-6pt}
\begin{align*}
& \MC{E}_{w,\sigma}(\MC{S},0,\infty) = \sum_{1\le i\le \kappa}\MC{E}_{w,\sigma}(\MC{S},\ell_i,r_i) \\
\le & \sum_{1\le i\le \kappa}\BIGP{\psi_{\sigma}\cdot\sum_{j\in \MB{J}_i} c_j} + \BIGP{3-\lambda}\cdot\sum_{1\le i\le \kappa}\MC{E}_{b,\sigma}(\MC{O},r_{i-1},r_i) \\
& \qquad \qquad \qquad \qquad \qquad \quad
+ 4\cdot\sum_{1\le u\le m}\MC{E}_{w}(\MC{O}, \ell_{u,opt}, r_{u,opt}) \\
\le & \BIGP{4-\lambda}\cdot\MC{E}_{b,\sigma}(\MC{O},0,\infty) + 4\cdot\MC{E}_{w}(\MC{O},0,\infty),
\end{align*}
where 
the last inequality comes from
the assumption that $\psi_{\sigma} \leq \psi_{b}$.
\qed
\end{proof}

\begin{theorem} \label{theorem-sparse-factor}
By setting $\lambda=1$, algorithm $\MC{S}$ computes a $4$-competitive schedule for any given set of jobs satisfying Condition~(\ref{eq-edf}) for the online power-minimizing scheduling problem.
\end{theorem}

\begin{proof}
The feasibility of $\MC{S}$ is guaranteed by Lemma~\ref{lemma-feasibility}.
For the competitive factor, let $\MC{E}(\MC{S})$ and $\MC{E}(\MC{O})$ denote the total
energy consumption for schedule $\MC{S}$ and $\MC{O}$, respectively.
First, we have $\MC{E}_{b}(\MC{S},0,\infty) \le \MC{E}_{b,\sigma}(\MC{O},0,\infty)$, since the workload to execute is the same in both $\MC{S}$ and $\MC{O}$.
Together with Lemma~\ref{lemma_sparse_energy_ratio}, we have
$\MC{E}(\MC{S}) 
= \MC{E}_{b}(\MC{S},0,\infty) + \MC{E}_{w,\sigma}(\MC{S},0,\infty) 
\le \BIGP{5-\lambda}\cdot\MC{E}_{b,\sigma}(\MC{O},0,\infty)+4\cdot\MC{E}_{w}(\MC{O},0,\infty) 
\le \BIGP{5-\lambda}\cdot\MC{E}(\MC{O})$,
provided that $2-\sqrt{3}\le\lambda\le 1$. 
By choosing $\lambda$ to be $1$, we have $\MC{E}(\MC{S}) \le 4\cdot\MC{E}(\MC{O})$ as claimed.
\qed
\end{proof}

\subsection{Jobs with Unit Execution Time}

We show that, when the execution times of the jobs are unit, we can benefit even more from the energy-efficient anchor with a properly chosen parameter $\lambda = 4-\sqrt{10}$.
The major difference is that, when the system is in urgency while $\MB{Q}^{t^*}_{forth}(t)$ is empty, i.e, processor $M_2$ is in 
standby,
we 
use a global earliest-deadline-first scheduling by executing two jobs on $M_1$ and $M_2$ instead of keeping one processor in standby, which
improves resource utilization.
Let $\MC{S}^\dagger$ denote the modified algorithm.
As a result, both processors will keep running until the global ready queue is empty.
Hence, we know that, when $\MC{P}_i = 2$ for an awaken-interval $I_i$ of $\MC{S}^\dagger$, 
the total standby time in $I_i$ is at most $\BIGP{1-\frac{\lambda}{2}}\MC{B}$.
The following two lemmas are the updated versions of Lemma~\ref{lemma_sparse_energy_bound} and Lemma~\ref{lemma_sparse_energy_ratio}, respectively.

\begin{lemma} \label{lemma_sparse_unit_energy_bound}
When the jobs have unit execution times, for each awaken-interval $I_i$, $1\le i\le k$, we have $\MC{E}_{w,\sigma}(\MC{S}^\dagger,\ell_i,r_i) \le \BIGP{3-\frac{\lambda}{2}}E_w$.
%
\end{lemma}

\begin{lemma} \label{lemma_sparse_unit_energy_ratio}
When $7-\sqrt{41} \le \lambda \le 4-\sqrt{10}$, we have
\begin{align*}
\MC{E}_{w,\sigma}(\MC{S}^\dagger,0,\infty) \le \BIGP{3-\frac{\lambda}{2}}\cdot\MC{E}_{b,\sigma}(\MC{O},0,\infty) + \BIGP{4-\frac{\lambda}{2}}\cdot\MC{E}_{w}(\MC{O},0,\infty).
\end{align*}
\end{lemma}

By Lemma~\ref{lemma_sparse_unit_energy_ratio} with $\lambda$ chosen to be $4-\sqrt{10}$, we get $\MC{E}(\MC{S}^\dagger) \le \BIGP{4-\frac{\lambda}{2}}\cdot\MC{E}(\MC{O}) < 3.59\cdot\MC{E}(\MC{O})$.
We have the following theorem.

\begin{theorem} \label{thm-unit-jobs}
By setting $\lambda = 4-\sqrt{10}$, algorithm $\MC{S}^\dagger$ computes a $3.59$-competitive schedule for any given set of jobs with execution time that satisfies Condition~(\ref{eq-edf}) for the online power-minimizing scheduling problem.
\end{theorem}



\subsection{Trading the Number of Processors with the Energy-Efficiency} \label{sec-trade-off}

We have shown how a stream of jobs satisfying Condition~(\ref{eq-edf}) can be scheduled online to obtain a $4$-competitive schedule which uses at most two processors.
By collecting the delayed jobs and bundling their execution, we can generalize the algorithm to a prescribed collection of job streams, each of which satisfies Condition~(\ref{eq-edf}), to allow a trade-off between the number of processors we use and the energy-efficiency of the resulting schedule.
%

%
Let $\MC{J}_1, \MC{J}_2, \ldots, \MC{J}_k$ be the given collection of job streams, where $\MC{J}_i$ satisfies Condition~(\ref{eq-edf}) for all $1\le i\le k$.
Below we show that, for any fixed $h$, $h > k$, how a $\BIGP{4\cdot\max\BIGBP{\CEIL{\frac{k}{h-k}},1}}$-competitive schedule which uses at most $h$ processors can be obtained.
First, if $h \ge 2k$, then we simply apply the algorithm $\MC{S}$ on every pair of the streams, i.e., on $\MC{J}_{2i}$ and $\MC{J}_{2i+1}$, for all $1\le i \le \frac{k}{2}$, and we get a $4$-competitive schedule.
For the case $k < h < 2k$, we divide $\MC{J}_1, \MC{J}_2, \ldots, \MC{J}_k$ into $h-k$ subsets such that each subset gets at most $\CEIL{k / (h-k)}$ streams. The $h$ processors are allocated in the following way. Each stream of jobs gets one processor, and the remaining $h-k$ processors are equally allocated to each of the $h-k$ subsets.
%
%
%
Then, the problem is reduced to 
the remaining case, $h = k+1$.

Below we describe how the case $h = k+1$ is handled.
Let $M_0, M_1, \ldots, M_k$ denote the $k+1$ processors to be managed, and let $\MB{Q}_0, \MB{Q}_1, \ldots, \MB{Q}_k$ be the corresponding subset of jobs that are scheduled to be executed on each processor.
We use the parameter $\lambda = 1$ to set the energy-efficient anchor for each job that arrives.
Let $W_0(t,t^\dagger) = \sum_{j \in \MB{Q}_0,d_j \le t^\dagger}c^\prime_j(t)$ denote the total remaining execution time of the jobs in $\MB{Q}_0$ whose deadlines are less than or equal to $t^\dagger$, for any time $t$ and any $t^\dagger$ with $t^\dagger \ge t$.
The algorithm 
works as follows. When a job $j \in \MC{J}_i$, $1\le i\le k$, arrives, 
if $M_i$ is on, then we add $j$ to $\MB{Q}_i$. Otherwise, we further check whether $W_0(t, t^\dagger) + c_j \le t^\dagger-t$ holds for all $t^\dagger \ge d_j$. If it does, then $j$ is added to $\MB{Q}_0$. Otherwise, we add $j$ to $\MB{Q}_i$ and turn on the processors $M_i$ and $M_0$ (if $M_0$ is off).
At any time $t$, we 
have the following cases to consider.

\begin{compactitem}
\item
First, regarding the further conditions to turn on $M_0$, 
%
if $M_0$ is off, and if the energy-efficient anchor of some job in $\MB{Q}_0$ is met or there exists some $t^\dagger \ge t$ such that $W_0(t,t^\dagger) \ge t^\dagger-t$, 
then we turn on $M_0$.

\item
Second, regarding the scheduling,
for each $0\le i\le k$ such that $M_i$ is on, we use the earliest-deadline-first principle to schedule the jobs of $\MB{Q}_i$ on $M_i$.

\item
Finally, regarding the conditions for turning off the processor, when $\MB{Q}_0$ becomes empty and $M_0$ has been turned on for at least $\MC{B}$ amount of time, we turn off $M_0$. For $1\le i\le k$, if $M_i$ is on, $\MB{Q}_i$ becomes empty, and $M_0$ is off, then we turn off $M_i$.
\end{compactitem}

\begin{theorem} \label{thm-multi-core-factor}
Given a collection of $k$ job sets, each satisfying Condition~(\ref{eq-edf}), we can compute a $\BIGP{4\cdot\max\BIGBP{\CEIL{\frac{k}{h-k}},1}}$-competitive schedule 
that uses at most $h$ processors, where $h >k$, for the online power-minimizing scheduling problem.
\end{theorem}

\section{Conclusion} \label{section_conclusion}

We conclude with a brief 
overview on future research directions.
%
With the advances of technology, racing-to-idle has become an important scheduling model for energy efficiency. 
%
%
For online settings, it still remains
open how a balance between (1) the extent to which the executions of the jobs should be postponed, and (2) the extent to which the system should speed up
for impulse arrivals, can be reached or even be traded with other parameters.
We believe that this will be a very interesting research direction to explore.



\bibliographystyle{plain}
\bibliography{energy_short_bib}


\newpage

\begin{appendix}




\section{Problem Lower Bound}


\begin{ap_lemma}{\ref{lemma:case1-lower}}
If the first job released by $\MC{A}$ is scheduled at time $t$ with $0 \leq t \le \BIGP{\frac{1}{2}-x}\MC{B}$, then we have
$$\frac{\MC{E}(\Pi)}{\MC{E}(\MC{O})} \ge 2+\frac{1}{2}x-{O}\BIGP{\frac{1}{k}} \qquad \text{for any $0 \le x \le \frac{2}{5}$}.$$
\end{ap_lemma}

\begin{proof}[Proof of Lemma~\ref{lemma:case1-lower}]
Let $m$ be the number of jobs $\MC{A}$ releases, including the first job at time $0$. Suppose that these $m$ jobs are indexed in an ascending order according to their arrival time.
By definition $a_1 = 0$ and $a_m \le \frac{3}{2}\MC{B}$.
  
Hence, the online algorithm $\Pi$ turns on a processor at least $m$ times, in which the energy consumption for switching on
is at least $m\cdot E_w$. Moreover, except for $\BIGP{m-1}\epsilon_0$
units of time when no processor is on, $\Pi$ keeps a
processor on for $\BIGP{\frac{3\MC{B}}{2}-t}-\BIGP{m-1}\epsilon_0$ time units
in the time interval $[t, \frac{3\MC{B}}{2}]$. Therefore $\MC{E}(\Pi)$
is at least $mE_w + \BIGP{\frac{3\MC{B}}{2}-t}\psi-\BIGP{m-1}\epsilon_0\psi$, which is at least
$mE_w + \BIGP{1+x}\MC{B}\psi-\BIGP{m-1}\psi$.

Then, we give the upper bound for the energy consumption of an
optimal offline schedule by constructing a specific feasible offline
schedule for these $m$ jobs.
Note that, the job $\MC{A}$ releases at time $0$ is flexible and can be scheduled at any time between $t$ and $\MC{B}-1$.
Specifically, we consider $4$ cases, depending on the value of $m$.
  
\emph{When $m = 1$}, we have only one job to execute and one feasible
schedule is to turn on the processor at any time $t<\MC{B}-1$
followed by immediately turning off. Therefore, for $m=1$, $\MC{E}(\MC{O}) \le E_w+\psi$, and
$$\frac{\MC{E}(\Pi)}{\MC{E}(\MC{O})} \ge \frac{k(2+x)}{k+1} = \BIGP{2+x}-{O}\BIGP{\frac{1}{k}}.$$
%

\emph{When $m= 2$}, there are two subcases. (1) If $a_2 \leq \MC{B}$, then the two jobs can be scheduled to run consecutively and the energy consumption is $E_w+\BIGP{1+\epsilon_0}\psi$.
(2) If $a_2 > \MC{B}$, the first job can be scheduled to run at time $\MC{B}-1$, and the processor idles from $\MC{B}$ to $a_2$ and executes
the second job at time $a_2$. The energy consumption is
$E_w+\BIGP{a_2-\MC{B}+\BIGP{1+\epsilon_0}}\psi$, which is maximized when $a_2 = \frac{3}{2}\MC{B}$. 
Therefore, when $m=2$, 
for both subcases we have $\MC{E}(\MC{O}) \le E_w + \BIGP{\frac{\MC{B}}{2}+2}\psi$, and
$$\frac{\MC{E}(\Pi)}{\MC{E}(\MC{O})} \ge \frac{k(3+x)-1}{\frac{3k}{2}+2} = 2+\frac{2}{3}x - {O}\BIGP{\frac{1}{k}}.$$

\emph{When $m = 3$}, there are two subcases as well. (1) If $a_3-a_2 > \MC{B}$, then $a_2 < \MC{B}$, and the first two jobs can be scheduled to run consecutively. The processor is 
off from $a_2+\epsilon_0$ to $a_3$.
The energy consumption for this subcase is $2E_w+\BIGP{1+2\epsilon_0}\psi$. (2) If $a_3-a_2 \leq \MC{B}$, we schedule the first job at any time between $\frac{1}{2}\MC{B}$ and $\MC{B}-1$. Disregarding $a_2 \le \MC{B}$ or not, 
the energy consumption for this subcase is at most $E_w+\BIGP{\MC{B}+\epsilon_0}\psi$. 
Therefore, when $m=3$, we have
$\MC{E}(\MC{O}) \le 2E_w + 3\psi$, and
$$\frac{\MC{E}(\Pi)}{\MC{E}(\MC{O})} \ge \frac{k(4+x)-2}{2k+3} = 2+\frac{1}{2}x - {O}\BIGP{\frac{1}{k}}.$$

\emph{When $m \ge 4$}, we keep the
processor on from time $t$ to time $\frac{3\MC{B}}{2}+\epsilon_0$. The energy consumption for the optimal offline schedule is at most $E_w +\BIGP{\frac{3\MC{B}}{2}-t+1}\psi$, and
$$\frac{\MC{E}(\Pi)}{\MC{E}(\MC{O})} \ge \frac{4k+\BIGP{\frac{3k}{2}-t}-3}{k+\frac{3k}{2}-t+1} = \frac{11k-2t-6}{5k-2t+2}.$$
Since $\MC{E}(\Pi) > \MC{E}(\MC{O}) \ge 0$ and $t \ge 0$, the above fraction achieves its minimum at $t = 0$.  Hence, when $m=4$, we have
$$\frac{\MC{E}(\Pi)}{\MC{E}(\MC{O})} \ge \frac{11k-6}{5k+2} = \frac{11}{5} - {O}\BIGP{\frac{1}{k}}.$$

Combining 
the above four cases, we get 
$$\frac{\MC{E}(\Pi)}{\MC{E}(\MC{O})} \ge \min\BIGBP{2+x, \: 2+\frac{2}{3}x, \: 2+\frac{1}{2}x, \: \frac{11}{5}} - {O}\BIGP{\frac{1}{k}} = 2+\frac{1}{2}x-{O}\BIGP{\frac{1}{k}},$$
when $k$ is large enough and $0 \le x \le \frac{2}{5}$.
\qed
\end{proof}

\medskip

\begin{ap_lemma}{\ref{lemma:case2-lower}}
If the first job released by the adversary is not executed until time $\BIGP{\frac{1}{2}-x}\MC{B}$, then for any $x,\eta,\chi \ge 0$ we have
$$\frac{\MC{E}(\Pi)}{\MC{E}(\MC{O})} \ge \min\BIGBP{\frac{3+x+\eta}{\frac{3}{2}+x},\frac{4+x+\eta+\chi}{2+x+\eta},\frac{5+x+\eta+\chi}{2+x+\eta+\chi}}-{O}\BIGP{\frac{1}{k}}.$$
\end{ap_lemma}

\begin{proof}[Proof of Lemma~\ref{lemma:case2-lower}]
  By our design, $\MC{A}$ releases exactly $\BIGP{\frac{1}{2}+x}\MC{B}$ jobs
  before time $\MC{B}$, forcing $\Pi$ to use at least two
  processors to have a feasible schedule.  Let $m_1$ be the number of
  the urgent unit jobs that $\MC{A}$ releases between time $\MC{B}$ and time
  $\BIGP{\frac{3}{2}+\eta}\MC{B}$.  If $m_1 = 0$, then $\Pi$
  keeps at least one processor in the standby mode till time
  $\BIGP{\frac{3}{2}+\eta}\MC{B}$, and the energy consumption is
  $\MC{E}(\Pi) \ge 2\cdot E_w + \BIGP{\frac{1}{2}+x}\MC{B}\cdot \psi
  + \BIGP{\frac{1}{2}+\eta}\MC{B}\cdot \psi$.  A feasible offline
  schedule can execute all these jobs consecutively on a processor
  from time $\BIGP{\frac{1}{2}-x}\MC{B}$ to $\MC{B}$, with energy
  consumption equal to $E_w + \BIGP{\frac{1}{2}+x}\MC{B}\cdot
  \psi$. Therefore, for $m_1=0$, we know that $\MC{E}(\MC{O}) \le E_w +
  \BIGP{\frac{1}{2}+x}\MC{B}\cdot \psi$, which also implies
  \begin{equation}
    \label{eq:case2:m1=0}
	\frac{\MC{E}(\Pi)}{\MC{E}(\MC{O})} \ge \frac{3+x+\eta}{\frac{3}{2}+x}.
  \end{equation}

  
  For the other case with $m_1 \ge 1$, the adversary $\MC{A}$ monitors
  the behavior of $\Pi$ till time
  $\BIGP{\frac{3}{2}+\eta+\chi}\MC{B}$. Let $m_2$ be the number of
  the urgent unit jobs $\MC{A}$ releases after time
  $\BIGP{\frac{3}{2}+\eta}\MC{B}$.  Then by an analogous argument,
  we have $\MC{E}(\Pi) \ge (2+m_1+m_2)\cdot E_w +
  \BIGP{\frac{1}{2}+x}\MC{B}\cdot \psi +
  \BIGP{\frac{1}{2}+\eta+\chi}\MC{B}\cdot \psi -
  \BIGP{m_1+m_2}\psi$.
  Note that this function is minimized when $m_1$ and $m_2$ achieve their minimums.

  A feasible offline schedule on a processor for $m_1 \ge 1$ is
  constructed by turning on the processor at time
  $\BIGP{\frac{1}{2}-x}\MC{B}$ and executing the coming jobs. If $m_2
  = 0$, then we turn off the processor at time
  $\BIGP{\frac{3}{2}+\eta}\MC{B}$, and, hence,
  $\MC{E}(\MC{O}) \leq E_w + \BIGP{1+x+\eta}\MC{B}\psi$. Otherwise, we turn the processor
  off at time $\BIGP{\frac{3}{2}+\eta+\chi}\MC{B}$, and, hence,  $\MC{E}(\MC{O}) \leq E_w +
  \BIGP{1+x+\eta+\chi}\MC{B}\psi$.
  Therefore,
  \begin{align}
    \frac{\MC{E}(\Pi)}{\MC{E}(\MC{O})} & \ge \min\left\{
    \begin{array}{l}
      \frac{(3+m_1+x+\eta+\chi)k - m_1}{\BIGP{2+x+\eta}\cdot k},\\
    \frac{(3+m_1+m_2+x+\eta+\chi)k - \BIGP{m_1+m_2}}{\BIGP{2+x+\eta+\chi}\cdot k}
    \end{array}
    \right\}\nonumber\\
    & \ge \min\BIGBP{\frac{4+x+\eta+\chi}{2+x+\eta},
      \frac{5+x+\eta+\chi}{2+x+\eta+\chi}}-{O}\BIGP{\frac{1}{k}} \label{eq:case2:m1>0}
  \end{align}
  By Eqs. ~(\ref{eq:case2:m1=0}) and (\ref{eq:case2:m1>0}), the lemma is proved. 
\qed
\end{proof}

\medskip

\begin{ap_theorem}{\ref{thm_sparse_lower_bound}}
The competitive factor of any online algorithm $\Pi$ for the online
power-minimizing 
scheduling problem is at least $2.06$.
\end{ap_theorem}

\begin{proof}[Proof of Theorem~\ref{thm_sparse_lower_bound}]
%
%
  By Lemmas~\ref{lemma:case1-lower} and \ref{lemma:case2-lower}, we know that
  \begin{align*}
  \frac{\MC{E}(\Pi)}{\MC{E}(\MC{O})} \geq \sup_{\scriptsize\begin{array}{c} 0\le x \le 2/5 \\ 0 \le \eta,\chi\end{array}}\min 
  \BIGLR{\{}{.}{
		\frac{4+x}{2}, \frac{3+x+\eta}{\frac{3}{2}+x}, 
		 \frac{4+x+\eta+\chi}{2+x+\eta}, } 
		\BIGLR{.}{\}}{ 
		\frac{5+x+\eta+\chi}{2+x+\eta+\chi}
    }-{O}\BIGP{\frac{1}{k}}.
  \end{align*}

To get an asymptotic sup-min of the four items inside the brace, observe that, as $\chi$ increases, the third item increases while the fourth item decreases. Therefore, the value of $\chi$ for which the last two items achieve their max-min can be solved.
In this way, we can solve $\eta$ and $x$ to get the asymptotic lower bound for $\MC{E}(\Pi) / {\MC{E}(\MC{O})}$.

By 
choosing $x$, $\eta$, and $\chi$ to be $0.1218$, $0.2206$, and $0.4852$, respectively, we get $\MC{E}(\Pi) \ge 2.06\cdot\MC{E}(\MC{O})$ when $k \gg 1$.
\qed
\end{proof}

%


\section{Online Scheduling}


\subsection{Common Approaches and Bad Examples}

\begin{ap_lemma}{\ref{lemma_ltr_lower_bound} }
The competitive factor of $\MC{L}$ can be arbitrarily bad.
Furthermore, even when $c_j = 1$ for all $j \in \MC{J}$, the competitive factor of $\MC{L}$ is still at least $6$.
\end{ap_lemma}

\begin{proof}[Proof of Lemma~\ref{lemma_ltr_lower_bound}]

For the first part of this lemma,
%
consider the following job set, $\MB{J}_{edf}$. Let $k$ be a non-negative integer. At time $0$, a job $j_1 = (0, 4k, 3k)$ arrives and is delayed. The system remains in standby until time $k$, at which $k-1$ unit jobs arrive to the system with deadline $4k-1$. 
As a result, $j_1$ will miss its deadline unless we turn on $k$ processors.
This shows that, the number of processors $\MC{L}$ will use for $\MB{J}_{edf}$ can be arbitrarily large, whereas the optimal schedule uses only one processor, and the competitive factor can therefore be arbitrarily large.
%

%
For the second part of this lemma, we define another job set $\MB{J}$.
Let $n>0$ be an even integer. The job set $\MB{J}$ consists of $\frac{3}{2}n$ jobs: 
$j_1,j_2,\ldots, j_n$, $j^\prime_1,j^\prime_2,\ldots,j^\prime_{n/2}$, where
\begin{align*}
\begin{cases}
j_i = \BIGP{(i-1)\MC{B}+2i,i\MC{B}+2i+1}, & \text{for each $1\le i\le n$, and} \\
j^\prime_k = \BIGP{2k\MC{B}+4k, 2k\MC{B}+4k+1}, & \text{for each $1\le k\le \frac{n}{2}$.}
\end{cases}
\end{align*}

Let $\MC{E}(\MC{L})$ and $\MC{E}(\MC{O})$ be the energy consumed by
$\MC{L}$ and the offline optimal schedule for the input instance
$\MB{J}$, respectively.  We prove this theorem by showing that
$\MC{E}(\MC{L}) $ is asymptotically no less than $6\cdot\MC{E}(\MC{O})$.

First we consider the behavior of the algorithm $\MC{L}$. 
At the beginning, $\MC{L}$ postpones the execution of $j_1$ until time
$\MC{B}+2$, for which it turns on a processor to execute $j_1$,
followed by staying in standby for $\MC{B}$ units of time until it 
turns off the processor
at time $2\MC{B}+3$.  At time $2\MC{B}+4$ when $\MC{L}$
will schedule $j_2$ to execute, 
an urgent job
$j^\prime_1$ arrives. As a result, $\MC{L}$
has to turn on another processor for it.  Then $\MC{L}$ keeps both
processors in standby followed by
turning them off
at time $3\MC{B}+5$.
The same pattern is repeated for the remaining $\frac{3n}{2}-3$ jobs.
See also Fig~\ref{fig_uni_tight_6} for an illustration.
For each repetition of the pattern, say, for jobs $j_{2k-1}, j_{2k}$, and $j^\prime_k$, $1\le k\le \frac{n}{2}$,
the energy consumption of $\MC{L}$ is at least $3\cdot E_w + 3\MC{B}\cdot \psi_{\sigma} = 6E_w$.
Therefore $\MC{E}(\MC{L}) \ge \frac{n}{2}\cdot 6E_w = 3nE_w$.

\begin{figure*}[htp]
\centering
\fbox{\includegraphics[scale=.8]{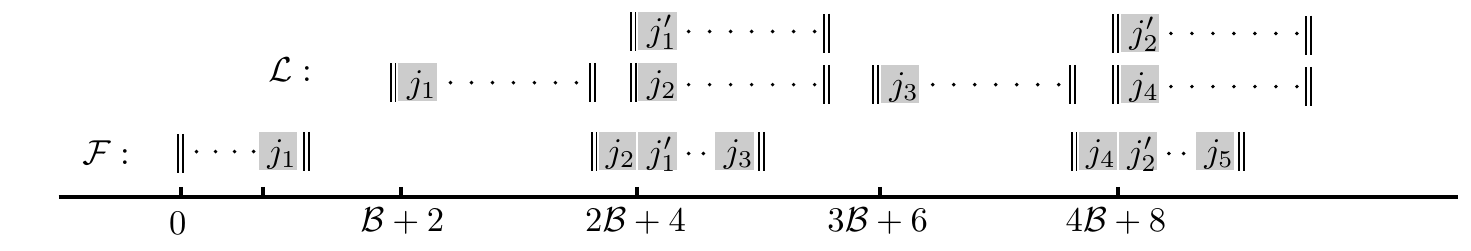}}
\caption{
Dotted horizontal lines represent the periods for which the processor is in the standby mode, and double vertical lines represent the turn-ons and turn-offs.}
\label{fig_uni_tight_6}
\end{figure*}

A feasible schedule for $\MB{J}$, denoted by $\MC{F}$, can be constructed as
follows.  For each $1\le k\le \frac{n}{2}$, we turn on a processor at
time $\max\BIGBP{0,\BIGP{2k-2}\MC{B}+4k-5}$ and turn off the processor
at time $\BIGP{2k-2}\MC{B}+4k-1$.
We schedule $j_{2k-1}$ at the time it arrives, which is at $\BIGP{2k-2}\MC{B}+4k-2$.
Then, $j_{2k}$ and $j^\prime_k$ are scheduled at
$2k\MC{B}+4k-1$ and $2k\MC{B}+4k$, respectively.
The energy consumption of this schedule, which is also an upper-bound
of 
$\MC{E}(\MC{O})$, is at most
$\BIGP{\frac{n}{2}+1}\BIGP{E_w+4\psi_{b}}$.
Therefore,
$$\frac{\MC{E}(\MC{L})}{\MC{E}(\MC{O})} \ge \frac{3nE_w}{\BIGP{\frac{n}{2}+1}\BIGP{E_w+4\psi_{b}}} = 6 - 
O\BIGP{\frac{1}{E_w+\psi_b}} - O\BIGP{\frac{1}{n}},$$
which converges to $6$ when $E_w \gg \psi_{b} \ge 1$ and $n \gg 1$.
\qed
\end{proof}

%


\subsection{The Analysis}

\begin{ap_lemma}{\ref{lemma-feasibility}}
Provided that Condition~(\ref{eq-edf}) from Lemma~\ref{lemma_offline_condition} holds for the input jobs, algorithm $\MC{S}$ always produces a feasible schedule.
\end{ap_lemma}

\begin{proof}[Proof of Lemma~\ref{lemma-feasibility}]
Consider any awaken-interval $I_i$, where $1\le i\le \kappa$. 
If the urgency flag is never set true within $I_i$, then we know that only $M_1$ is turned on during $I_i$, and we have $W(\ell_i, t) < t-\ell_i$ for all $t > \ell_i$ at time $\ell_i$ when processor $M_1$ is turned on.
Moreover, by the design of $\MC{S}$, once $M_1$ is on, it keeps executing jobs using the EDF principle until the queue $\MB{Q}(t)$ becomes empty.
Therefore by Lemma~\ref{lemma_offline_condition}, the schedule $\MC{S}$ produces during $I_i$ is feasible.
On the other hand, if the urgency flag is set true at some moment $t^*$ between $\ell_i$ and $r_i$, then by a similar argument, we know that the schedules $\MC{S}$ produces for the subsets $\MB{Q}^{t^*}_{proc}(t)$ and $\MB{Q}^{t^*}_{forth}(t)$ are feasible on $M_1$ and $M_2$, respectively, for all $\ell_i \le t \le r_i$.
%
%
\qed
\end{proof}

\begin{ap_lemma}{\ref{lemma-awaken-interval-property}}
For each $i$, $1\le i\le \kappa$, we have
\begin{itemize}
	\item
		$\MC{P}_i \le 2$.
		
		\smallskip
		
	\item
		If $\MC{P}_i = 2$, then the amount of workload that arrives after $r_{i-1}$ and has to be done before $r_i$ is at least $\lambda\MC{B}$.
\end{itemize}
\end{ap_lemma}

\begin{proof}[Proof of Lemma~\ref{lemma-awaken-interval-property}]
If the 
urgency flag is never set true in $I_i$, then $\MC{P}_i = 1$. Otherwise, suppose the system enters and leaves urgency at time $t^*$ and $t^\dagger$, respectively. By our scheme, $\MB{Q}^{t^*}_{proc}(t^\dagger)$ is empty,
and $W(t^\dagger,t) \le t-t^\dagger$ for all $t \ge t^\dagger$ by our assumption on Condition~(\ref{eq-edf}). Therefore $S$ will not turn on $M_1$ again before $I_i$ ends. This shows $\MC{P}_i \le 2$.

\smallskip

Suppose that $\MC{S}$ turns on $M_2$ at time $t^*$, and $t^\dagger$ is the corresponding moment for which $W(t^*,t^\dagger)> t^\dagger-t^*$.
%
%
Then by the design of $\MC{S}$, we know that
the ready queue is never empty between time $\ell_i$ and $t^*$, for otherwise it will contradict Condition~(\ref{eq-edf})
since $W(t^*,t^\dagger)> t^\dagger-t^*$. 
Therefore the system stays in the running mode and never enters the standby mode between time $\ell_i$ and $t^*$.

Since we use EDF principle, the workload that has to be done before $t^\dagger$ is hence strictly greater than $t^\dagger-\ell_i$. Namely, among these jobs, at least one, say, $j$, has arrival time earlier than $\ell_i$. That is, $r_{i-1} \le a_j < \ell_i$. Therefore $d_j \ge \ell_i + \lambda\MC{B}$; otherwise $\MC{S}$ would have turned on $M_1$ earlier than $\ell_i$ 
when the energy-efficient anchor of $j$ is met. Therefore the total workload that is done within $I_i$ is at least $t^\dagger-\ell_i \ge d_j - \ell_i \ge \lambda\MC{B}$.
\qed
\end{proof}

\medskip

\begin{ap_lemma}{\ref{lemma_sparse_sleep_vs_working}}
For any $i$, $1\le i\le \kappa$, if 
both $r_{i-1}$ and $\ell_i$ are contained in some sleep interval $U_u = (\ell_{u,opt}, r_{u,opt})$ of $\MC{O}$,
then $\MC{P}_i = 1$ and $r_i - r_{u,opt} \ge \BIGP{1-\lambda}\MC{B}$.
\end{ap_lemma}

\begin{proof}[Proof of Lemma~\ref{lemma_sparse_sleep_vs_working}]
Since $I_{i-1}$ ends when the ready queue is empty, the arrival time of each job in $\MB{J}_i$ is later than $r_{i-1}$. 
Therefore $\MC{P}_i = 1$ as $\MC{O}$ is a feasible uni-processor schedule which turns on the system to process jobs no earlier than $\ell_i$.

Let $j$ be the job scheduled to execute at time $\ell_i$. Since we use EDF principle, $d_{j^\prime} \ge d_j \ge r_{u,opt}$ for all $j^\prime \in \MB{Q}(\ell_i)$.
There are two cases to set $\ell_i$ in $\MC{S}$: (1) 
the jobs
in $\MB{Q}(\ell_i)$ will miss the deadline if the system is not turned
on for processing at time $\ell_i$, and (2) the energy-efficient anchor of a certain job is $\ell_i$. For the first case, the sleep interval $U_u$ cannot 
contain $\ell_i$.
Hence, when $\ell_i$ 
is contained in 
$U_u$, we have
$d_j - \ell_i \le \lambda\MC{B}$. Therefore,
$r_i - r_{u,opt} \ge r_i - d_j \ge r_i - \BIGP{\lambda\MC{B}+\ell_i} \ge \BIGP{1-\lambda}\MC{B}$.
\qed
\end{proof}

\medskip

\begin{ap_lemma}{\ref{lemma_sparse_sleep_n_awaken}}
Each sleep interval of $\MC{O}$ intersects with at most two awaken-intervals of $\MC{S}$. In particular, each of $U_1$ and $U_m$ intersects with at most one.
\end{ap_lemma}

\begin{proof}[Proof of Lemma~\ref{lemma_sparse_sleep_n_awaken}]
This lemma follows directly from Lemma~\ref{lemma_sparse_sleep_vs_working}. For $1\le i\le\kappa-2$, if $r_i$ and $\ell_{i+1}$ are contained in a sleep interval $U_u = (\ell_{u,opt}, r_{u,opt})$ of $\MC{O}$, then $r_{i+1} - r_{u,opt} \ge \BIGP{1-\lambda}\MC{B} > 0$, which implies that $\ell_{i+2}$ will not be contained in $U_u$.
\qed
\end{proof}

\medskip

\begin{ap_lemma}{\ref{lemma_sparse_energy_bound}}
  For each awaken-interval $I_i$, $1\le i\le \kappa$ and $\lambda\leq 1$, 
  \begin{align*}
    \MC{E}_{w,\sigma}(\MC{S},\ell_i,r_i) \le \psi_{\sigma}\BIGP{\sum_{j\in \MB{J}_i} c_j} + (3-\lambda)E_w. \\
    \text{Furthermore, if } \MC{P}_i = 1,\text{ then } \MC{E}_{w,\sigma}(\MC{S},\ell_i,r_i) \le 2E_w.
  \end{align*}
\end{ap_lemma}

\begin{proof}[Proof of Lemma~\ref{lemma_sparse_energy_bound}]
  If $\MC{P}_i = 1$, then the total idle time in $I_i$ is at most $\MC{B}$, since in our design, $\MC{S}$ will 
turn off the processor 
immediately once $\MC{Q}(t)$ is empty when $t \ge \ell_i+\MC{B}$. Therefore $\MC{E}_{w,\sigma}(\MC{S},\ell_i,r_i) \le E_w + \psi_{\sigma}\cdot\MC{B} = 2E_w$, which is at most $\psi_{\sigma}\cdot\BIGP{\sum_{j\in \MB{J}_i} c_j} + \BIGP{3-\lambda}\MC{B}$ when $\lambda \le 1$.
  
If $\MC{P}_i = 2$, only processor $M_2$ could be in the standby mode according to our design.
For the total standby time in $I_i$, we have two cases to consider: 
(1) If the system exits urgency after $\ell_i+\MC{B}$, then the total standby time on processor $M_2$ is upper-bounded by the workload on processor $M_1$, which is at most $\sum_{j\in \MB{J}_i} c_j$. 
(2) If the system exits urgency before $\ell_i+\MC{B}$, the total standby time is at most $\MC{B}$, which is at most $\sum_{j\in \MB{J}_i} c_j + \BIGP{1-\lambda}\MC{B}$ since $\sum_{j\in \MB{J}_i} c_j \ge \lambda\MC{B}$ by Lemma~\ref{lemma-awaken-interval-property}.
In both cases, the total standby time is at most $\sum_{j\in \MB{J}_i} c_j + \BIGP{1-\lambda}\MC{B}$.
Therefore $\MC{E}_{w,\sigma}(\MC{S},\ell_i,r_i) \le 2E_w + \psi_{\sigma}\cdot\BIGP{\sum_{j\in \MB{J}_i} c_j + \BIGP{1-\lambda}\MC{B}}$.
By Lemma~\ref{lemma-awaken-interval-property}, 
we only have to consider $\MC{P}_i=1$ and $\MC{P}_i = 2$, and this concludes the proof.
\qed
\end{proof}

\medskip

%



\subsection{Jobs with Unit Execution Time}

\begin{ap_lemma}{\ref{lemma_sparse_unit_energy_bound}}
When the jobs have unit execution times, for each awaken-interval $I_i$, $1\le i\le k$, we have $\MC{E}_{w,\sigma}(\MC{S}^\dagger,\ell_i,r_i) \le \BIGP{3-\frac{\lambda}{2}}E_w$.
%
\end{ap_lemma}

\begin{proof}[Proof of Lemma~\ref{lemma_sparse_unit_energy_bound}]
Since the total time in $I_i$ is at most $\BIGP{1-\frac{\lambda}{2}}\MC{B}$, we have $$\MC{E}_{w,\sigma}(\MC{S}^\dagger,\ell_i,r_i) \le 2E_w + \BIGP{1-\frac{\lambda}{2}}\MC{B}\cdot\psi_{\sigma} \le \BIGP{3-\frac{\lambda}{2}}E_w.$$
\qed
\end{proof}

\medskip

\begin{ap_lemma}{\ref{lemma_sparse_unit_energy_ratio}}
When $7-\sqrt{41} \le \lambda \le 4-\sqrt{10}$, we have
\begin{align*}
\MC{E}_{w,\sigma}(\MC{S}^\dagger,0,\infty) \le \BIGP{3-\frac{\lambda}{2}}\cdot\MC{E}_{b,\sigma}(\MC{O},0,\infty) + \BIGP{4-\frac{\lambda}{2}}\cdot\MC{E}_{w}(\MC{O},0,\infty).
\end{align*}
\end{ap_lemma}

\begin{proof}[Proof of Lemma~\ref{lemma_sparse_unit_energy_ratio}]
We sketch only the main differences from the proof of Lemma~\ref{lemma_sparse_energy_ratio} and adopt the remaining detail.
For any awaken-interval $I_i$ of $\MC{S}^\dagger$,
\begin{itemize}
	\item
		If $I_i$ does not intersect with any sleep interval of $\MC{S}^\dagger$, then 
		$$\MC{E}_{w,\sigma}(\MC{S}^\dagger,\ell_i,r_i) \le \BIGP{3-\frac{\lambda}{2}}\cdot E_w \le \BIGP{3-\frac{\lambda}{2}}\cdot\MC{E}_{b,\sigma}(\MC{O},\ell_i,r_i).$$
		
	\item
		If $I_i$ intersects with a sleep interval $U_u = (\ell_{u,opt},r_{u,opt})$ and
		$r_i$ is contained in $U_u$,
		we have two subcases.
		(a) If $\MC{P}_i=1$, then $$\MC{E}_{w,\sigma}(\MC{S}^\dagger,\ell_i,r_i) \le 2E_w = 2\MC{E}_{w}(\MC{O},\ell_{u,opt},r_{u,opt}).$$
		
		(b) If $\MC{P}_i = 2$, then $\MC{E}_{b,\sigma}(\MC{O},r_{i-1},\ell_{u,opt}) 
		\ge \lambda E_w$ by Lemma~\ref{lemma-awaken-interval-property}. 
		When $7-\sqrt{41}\le \lambda\le 1$, we have 
		$$\BIGP{1-\lambda}E_w \le \lambda\BIGP{3-\frac{\lambda}{2}}E_w.$$
		%
		Therefore, 
		\begin{align*}
		\MC{E}_{w,\sigma}(\MC{S}^\dagger,\ell_i,r_i) & \le \BIGP{1-\frac{\lambda}{2}}E_w + 2E_w \\
		& \le \BIGP{3-\frac{\lambda}{2}}\MC{E}_{b,\sigma}(\MC{O},r_{i-1},r_i) + 2\MC{E}_{w}(\MC{O},\ell_{u,opt},r_{u,opt}).
		\end{align*}
		
	\item
		If $I_i$ intersects with exactly one sleep interval $U_u = (\ell_{u,opt}, r_{u,opt})$ and
		$\ell_i$ is contained in $U_u$,
		then by Lemma~\ref{lemma_sparse_sleep_vs_working}, $\MC{P}_i = 1$ and
		$\MC{E}_{b,\sigma}(\MC{O},r_{u,opt},r_i) \ge 
		\BIGP{1-\lambda}E_w.$
		When $0\le \lambda \le 4-\sqrt{10}$, we have 
		$$\frac{\lambda}{2}E_w \le \BIGP{1-\lambda}\BIGP{3-\frac{\lambda}{2}}E_w.$$
		Therefore, 
		\begin{align*}
		 &\MC{E}_{w,\sigma}(\MC{S}^\dagger,\ell_i,r_i) \le 2E_w = \frac{\lambda}{2}E_2 + \BIGP{2-\frac{\lambda}{2}}E_w \\
         \le \enskip &
         \BIGP{3-\frac{\lambda}{2}}\MC{E}_{b,\sigma}(\MC{O},r_{u,opt},r_i) 
		 + \BIGP{2-\frac{\lambda}{2}}\MC{E}_{w}(\MC{O},\ell_{u,opt},r_{u,opt}).           
		\end{align*}
\end{itemize}

From the above three cases, when $7-\sqrt{41} \le \lambda \le
4-\sqrt{10}$, 
we have
\begin{align*}
\MC{E}_{w,\sigma}(\MC{S}^\dagger,0,\infty) 
\le \BIGP{3-\frac{\lambda}{2}}\cdot\MC{E}_{b,\sigma}(\MC{O},0,\infty) + \BIGP{4-\frac{\lambda}{2}}\cdot\MC{E}_{w}(\MC{O},0,\infty).
\end{align*}
\qed
\end{proof}

%


%




\subsection{Trading the Number of Processors with the Energy-Efficiency} 

\begin{ap_theorem}{\ref{thm-multi-core-factor}}
Given a collection of $k$ job sets, each satisfying Condition~(\ref{eq-edf}), we can compute a $\BIGP{4\cdot\max\BIGBP{\CEIL{\frac{k}{h-k}},1}}$-competitive schedule which uses at most $h$ processors, where $h >k$, for the online power-minimizing scheduling problem.
\end{ap_theorem}

\begin{proof}[Proof of Theorem~\ref{thm-multi-core-factor}]
The feasibility of the modified algorithm follows directly from our assumption on Condition~(\ref{eq-edf}) and Lemma~\ref{lemma_offline_condition}.

For the competitive factor of the algorithm, it suffices to consider the case $k<h<2k$. Consider the partition of the collection of job streams. Without loss of generality, let $\MC{J}_1, \MC{J}_2, \ldots, \MC{J}_{\CEIL{k/(h-k)}}$ be one subset, and let $M_0, M_1, \ldots, M_{\CEIL{k/(h-k)}}$ be the processors allocated to these job streams.
Consider any specific stream $\MC{J}_i$, $1\le i\le \CEIL{k/(h-k)}$, and the two processors $M_0$ and $M_i$. Observe that the lemmas presented in \S~\ref{sec-sparse-analysis} still holds for $M_0$ and $M_i$.
Therefore the energy consumed by $M_0$ and $M_i$ is bounded by four times the energy consumption of any offline optimal schedule for $\MC{J}_1, \MC{J}_2, \ldots, \MC{J}_{\CEIL{k/(h-k)}}$.
Since we have exactly $\CEIL{k/(h-k)}$ such pairs, each unit of energy consumption in the offline optimal schedule, including switch-on operations, energy for executing jobs, and energy for standby, is charged at most $\CEIL{k/(h-k)}$ times. Therefore, by taking all the pairs $M_0$ and $M_i$ into account, we get the factor $4\cdot\CEIL{\frac{k}{h-k}}$ as claimed.
\qed
\end{proof}

\end{appendix}


\end{document}